\documentclass[12pt,english]{article}
\usepackage[T1]{fontenc}
\usepackage[latin9]{inputenc}
\usepackage[letterpaper]{geometry}
\usepackage{float}
\usepackage{url}
\usepackage{amsthm}
\usepackage{amsmath}
\usepackage{amssymb}

\makeatletter

\newcommand{\lyxline}[1][1pt]{%
  \par\noindent%
  \rule[.5ex]{\linewidth}{#1}\par}
\providecommand{\tabularnewline}{\\}
\floatstyle{ruled}
\newfloat{algorithm}{tbp}{loa}
\floatname{algorithm}{Algorithm}

\theoremstyle{plain}
\newtheorem{thm}{Theorem}
  \theoremstyle{plain}
  \newtheorem{prop}[thm]{Proposition}
  \theoremstyle{remark}
  \newtheorem*{rem*}{Remark}

\usepackage{amsmath}
\usepackage{amssymb}
\usepackage{url}

\makeatother

\usepackage{babel}

\begin{document}

\title{Subtotal ordering -- a pedagogically advantageous algorithm for computing
total degree reverse lexicographic order}

\author{David R. Stoutemyer%
\thanks{dstout at hawaii dot edu%
}}
\maketitle
\begin{abstract}
Total degree reverse lexicographic order is currently generally regarded
as most often fastest for computing Gröbner bases. This article describes
an alternate less mysterious algorithm for computing this order using
exponent subtotals and describes why it should be very nearly the
same speed the traditional algorithm, all other things being equal.
However, experimental evidence suggests that subtotal order is actually
slightly faster for the \textsl{Mathematica}\textsuperscript{®} Gröbner
basis implementation more often than not. This is probably because
the weight vectors associated with the natural subtotal weight matrix
and with the usual total degree reverse lexicographic weight matrix
are different, and \textsl{Mathematica} also uses those the corresponding
weight vectors to help select successive S polynomials and divisor
polynomials: Those selection heuristics appear to work slightly better
more often with subtotal weight vectors.

However, the most important advantage of exponent subtotals is pedagogical.
It is easier to understand than the total degree reverse lexicographic
algorithm, and it is more evident why the resulting order is often
the fastest known order for computing Gröbner bases.
\end{abstract}

\section{Introduction\label{sec:Introduction}}

Total degree reverse lexicographic order (degRevLex) is currently
generally regarded as most often fastest for computing Gröbner bases..%
\footnote{The adjective {}``graded'' is sometimes used instead of {}``total
degree''.%
} This order is usually determined by Algorithm 1. However, as indicated
in the comments therein, this algorithm tends to be mystifying and
therefore difficult to recall:

\begin{algorithm}[H]

\caption{degRevLex order of two exponent vectors}

\textbf{Given}: Nonnegative integer exponent vectors $\boldsymbol{\alpha}=\left[\alpha_{1},\alpha_{2},\ldots,\alpha_{n}\right]$,
$\boldsymbol{\beta}=\left[\beta_{1},\beta_{2},\ldots,\beta_{n}\right]$.

\textbf{Returns}: One of {}``$\prec$'', {}``='', or {}``$\succ$''
according to the degRevLex order between power products $z_{1}^{\alpha_{1}}z_{2}^{\alpha_{2}}\cdots z_{n}^{\alpha_{n}}$
and $z_{1}^{\beta_{1}}z_{2}^{\beta_{2}}\cdots z_{n}^{\beta_{n}}$
with indeterminate order $z_{1}\succ z_{2}\succ\cdots\succ z_{n}$.

\lyxline{\normalsize}

$a\leftarrow\alpha_{1}$;

$b\leftarrow\beta_{1}$;

for $k\leftarrow2$ to $n$ do

$\quad a\leftarrow a+\alpha_{k}$;

$\quad b\leftarrow b+\beta_{k}$; \ end for;

if $a<b$, then return {}``$\prec$'';

if $a>b$, then return {}``$\succ$''; \quad{}\quad{}\quad{}\qquad{}\qquad{}/{*}
: \textsl{So far this all makes sense}. {*}/

for $k\leftarrow n$ to 1 by $-1$ do

$\quad$if $\alpha_{k}>\beta_{k}$, then return {}``$\prec$'';
\quad{}\qquad{}\qquad{}/{*} : \textsl{Whoa! Why is} {}``>'' \textsl{matched
with} {}``$\prec$''?

$\quad$if $\alpha_{k}<\beta_{k}$, then return {}``$\succ$'';
\ end for; \quad{}/{*} : \textsl{These must be typographical errors!}''

return {}``='';\qquad{}\qquad{}\qquad{}\qquad{}\qquad{}\qquad{}\qquad{}/{*}
: \textsl{At least this step makes sense!} {*}/
\end{algorithm}

This is the order that Buchberger in his Ph.D. dissertation and the
order that Gröbner always used when discussing multivariate polynomials
(Buchberger, personal communication).%
\footnote{Trinks \cite{Trinks} made the extraordinarily useful contribution
of introducing the idea of admissible orderings and the particularly
useful alternative example of lexicographic ordering.%
}

The next fastest widely discussed order is total degree lexicographic
order, and it too orders primarily by total degree. Therefore clearly
total degree is very important for speed, and that makes sense because
if we are iteratively annihilating terms with the largest total degrees,
then the degree of each variable can't increase beyond that total
degree.

Consequently, it seems plausible that it would be more consistent
to break total-degree ties with the sum of the degrees of all but
the least main variable, then break those ties with the sum of the
degrees of all but the two least main variables, and so on, as described
in Algorithm 2.%
\begin{algorithm}[H]
\caption{subtotal order of two exponent vectors}

\textbf{Given}: Nonnegative integer exponent vectors $\boldsymbol{\alpha}=\left[\alpha_{1},\alpha_{2},\ldots,\alpha_{n}\right]$
, $\boldsymbol{\beta}=\left[\beta_{1},\beta_{2},\ldots,\beta_{n}\right]$.

\textbf{Returns}: One of {}``$\prec$'', {}``='', or {}``$\succ$''
according to the subtotal order between power products $z_{1}^{\alpha_{1}}z_{2}^{\alpha_{2}}\cdots z_{n}^{\alpha_{n}}$
and $z_{1}^{\beta_{1}}z_{2}^{\beta_{2}}\cdots z_{n}^{\beta_{n}}$
with indeterminate order $z_{1}\succ z_{2}\succ\cdots\succ z_{n}$.

\lyxline{\normalsize}

$A_{1}\leftarrow\alpha_{1}$;

$B_{1}\leftarrow\beta_{1}$;

for $k\leftarrow2$ to $n$ do

$\quad A_{k}\leftarrow A_{k-1}+\alpha_{k}$;

$\quad B_{k}\leftarrow B_{k-1}+\beta_{k}$; \ end for;

for $k\leftarrow n$ to 1 by $-1$ do

$\quad$if $A_{k}>\beta_{k}$ then return {}``$\succ$''; \quad{}\qquad{}\qquad{}/{*}
: \textsl{I have a better feeling about this!} {*}/

$\quad$if $A_{k}<\beta_{k}$ then return {}``$\prec$''; \ end
for;

return {}``='';
\end{algorithm}

Algorithms 1 and 2 both do $2n$ additions followed by up to $n$
comparisons, with the same looping costs.

Both algorithms assume that the variables have been extracted from
the power products, which requires that all power products have the
same number of exponents even if some of these exponents are 0. If
instead the power products have only non-zero exponents, and therefore
also contain variables to indicate the base for each exponent, then
the algorithms are slightly different but both of them still have
the same complexity as each other.

Section \ref{sec:Weight-matrices-and-vectors} discusses weight matrices,
weight vectors, and their implications for subtotal \emph{versus}
degRevLex order. Section \ref{sec:Experimental-procedures-and} describes
the experimental procedures for comparing these two orders and the
results of those comparisons, with conclusions in Section \ref{sec:Conclusions}.

\section{\label{sec:Weight-matrices-and-vectors}Weight matrices and weight
vectors}

When I first thought of subtotal order, I wondered if subtotal order
would be faster than degRevLex order. Therefore I wanted a quick way
to compare them experimentally without having to implement my own
Gröbner basis algorithm or learn the intricacies of an existing one
to modify it. Fortunately some computer algebra systems permit users
to specify their own ordering merely by providing a non-singular real
square \textsl{weight matrix} $W$: To compare exponent vectors $\boldsymbol{\alpha}$
and $\boldsymbol{\beta}$ for ordering, we lexicographically compare
corresponding \textsl{weight vectors}\[
\begin{array}{ccccc}
\boldsymbol{w}\left(\boldsymbol{\alpha}\right) & \leftarrow & \left[w_{1}\left(\boldsymbol{\alpha}\right),w_{2}\left(\boldsymbol{\alpha}\right),\ldots\right] & \leftarrow & W\cdot\boldsymbol{\alpha}^{T},\\
\boldsymbol{w}\left(\boldsymbol{\beta}\right) & \leftarrow & \left[w_{1}\left(\boldsymbol{\beta}\right),w_{2}\left(\boldsymbol{\beta}\right),\ldots\right] & \leftarrow & W\cdot\boldsymbol{\beta}^{T}.\end{array}\]

The weight matrix corresponding to algorithm 2 is\begin{equation}
W_{\mathrm{sub}}=\left(\begin{array}{ccccc}
1 & 1 & \ldots & 1 & 1\\
1 & 1 & \ldots & 1 & 0\\
\vdots & \vdots & \ddots & \vdots & \vdots\\
1 & 1 & \ldots & 0 & 0\\
1 & 0 & \ldots & 0 & 0\end{array}\right).\label{eq:Wsub}\end{equation}

\begin{prop}
The ordering specified by $W_{\mathrm{sub}}$ is an admissible ordering.\end{prop}
\begin{proof}
The first element of each column in $W_{\mathrm{sub}}$ is positive.

Also, matrix $W_{\mathrm{sub}}$ is non-singular because the columns,
hence the variables, can be permuted into an upper triangular matrix
having 1 for every diagonal element, and the determinant of a square
upper triangular matrix is the product of its diagonal elements, which
consequently is non-zero, making the permuted $W_{\mathrm{sub}}$,
hence also $W_{\mathrm{sub}}$ non-singular. These are sufficient
conditions for an admissible ordering.
\end{proof}
Every admissible term order can be associated with a weight matrix,
and the one that is usually associated with algorithm 1 for degRevLex
ordering is\begin{equation}
W_{\mathrm{degRevLex}}=\left(\begin{array}{ccccc}
1 & 1 & \ldots & 1 & 1\\
0 & 0 & \ldots & 0 & -1\\
0 & 0 & \ldots & -1 & 0\\
\vdots & \vdots &  & \vdots & \vdots\\
0 & -1 & \ldots & 0 & 0\end{array}\right).\label{eq:Wgrevlev}\end{equation}

There are at least two ways to use weight matrices in an implementation:
\begin{enumerate}
\item Do a matrix-vector multiplication for a power product every time we
want to order it relative to another power product. To save some time,
we could interleave the matrix-vector multiplication with the comparison
of successive weight vector components to quit as soon as a difference
is detected.
\item Do matrix-vector multiplications only for the given \textsl{input}
polynomials, store the weight vectors in parallel with the corresponding
exponent vectors, then

\begin{enumerate}
\item whenever two power products are multiplied, compute the weight vector
of their product as the sum of the two weight vectors, and
\item whenever two power products are divided, compute the weight vector
of their quotient as the difference of the two weight vectors.
\end{enumerate}
\end{enumerate}
At the expense of some additional storage space, the second choice
is clearly much faster, because the number of power product comparisons
during computation of a Gröbner basis is usually far more than the
number of power products in the given polynomials.

\textsl{Mathematica} accepts weight matrices, and it uses the second
of these two methods -- even for the implemented built-in orderings,
for which it uses essentially Algorithm 1 rather than $W_{\mathrm{degRevLex}}$
to initialize the weight vectors in the case of degRevLex ordering.
Thus we can expect built-in degRevLex order for that implementation
to be slightly faster than providing $W_{\mathrm{degRevLex}}$, but
usually not dramatically so.

Nonetheless, to make the comparison fair -- more as if comparing built-in
degRevLex ordering to built-in subtotal ordering with initialization
via Algorithm 2 -- I started timing a few examples done with both
$W_{\mathrm{sub}}$ and $W_{\mathrm{degRevLex}}$. I soon noticed
that both matrices always gave the same Gröbner basis. This could
have been a coincidence. However:
\begin{prop}
Subtotal ordering is equivalent to degRevLex.\end{prop}
\begin{proof}
The first rows of $W_{\mathrm{sub}}$ and $W_{\mathrm{degRevLex}}$
are identical. The second row of $W_{\mathrm{sub}}$ given by formula
(\ref{eq:Wsub}) is the second row of $W_{\mathrm{degRevLex}}$ given
by formula (\ref{eq:Wgrevlev}) plus the first row of $W_{\mathrm{degRevLex}}$.
The third row of $W_{\mathrm{sub}}$ is the first row of $W_{\mathrm{degRevLex}}$
plus the second and third rows of $W_{\mathrm{degRevLex}}$, and so
on. Thus $W_{\mathrm{degRevLex}}$ can be transformed into $W_{\mathrm{sub}}$
by a sequence of adding positive multiples of rows to rows below them.
This is a sufficient condition for orderings specified by two weight
matrices to be equivalent.\end{proof}
\begin{rem*}
After noticing the identical Gröbner bases I discovered that Chee-Keng
Yap \cite{Yap} describes the computation of degRevLex ordering by
Algorithm 1, but lists $W_{\mathrm{sub}}$ rather than the usual $W_{\mathrm{degRevLex}}$
as a weight matrix for degRevLex ordering. He didn't explain why,
but I suspect that he had the idea of subtotals too.
\end{rem*}
With weight matrices $W_{\mathrm{degRevLex}}$ and $W_{\mathrm{sub}}$
inducing the same ordering, we should expect their speed ratios to
be close to 1.0.

Section \ref{sec:Experimental-procedures-and} indicates that this
is often so, but not always, then explains why there are exceptions
in the case of \textsl{Mathematica}.

\section{\label{sec:Experimental-procedures-and}Experimental procedures and
results}

The discussion in Sections \ref{sec:Introduction} and \ref{sec:Weight-matrices-and-vectors}
suggests that subtotal order should be very nearly the same speed
as degRevLex order -- regardless of whether they are both determined
directly from exponent vectors or both determined via weight matrices
and perhaps also weight vectors. This section describes experimental
procedures and results that test this hypothesis.

\subsection{Experimental procedures}

A \textsl{Mathematica} function invocation of the form\begin{multline*}
\mathtt{TimeConstrained\,}[\mathtt{Timing\,}[\\
\mathtt{GroebnerBasis\,}[\left\{ \mathit{polynomial}_{1},\mathit{polynomial}_{2},\ldots\right\} ,\,\left\{ \mathit{variable}_{1},\mathit{variable}_{2},\ldots\right\} ,\\
\mathtt{Sort}\rightarrow\mathtt{True},\mathtt{\: MonomialOrder}\rightarrow\mathtt{DegreeReverseLexicographic}];],\\
\mathit{maximumSeconds}]\end{multline*}
computes a degRevLex Gröbner basis of the polynomials after heuristically
reordering the variables for speed, then displays only the computing
time -- or displays $\mathtt{\$Aborted}$ if \emph{maximumSeconds}
is exceeded.

The rewrite rule\[
\mathtt{SubtotalWeightMatrix}\left[n\_\right]:=\mathtt{Table}\left[\mathtt{Table}\left[\mathtt{If}\left[i\leq n-j+1,\,1,\,0\right],\, j,\, n\right],\, i,\, n\right]\]
defines a function that returns an $n$ by $n$ weight matrix for
subtotal ordering.

The rewrite rule\begin{multline*}
\mathtt{DegRevLexWeightMatrix\,}[n\_]:=\\
\mathtt{Table\,}[\mathtt{Table\,}[\mathtt{If}\,[i==1,1,\mathtt{If\,}[i==n+2-j,-1,0]],\left\{ j,n\right\} ],\left\{ i,n\right\} ]\end{multline*}
defines a function that returns an $n$ by $n$ weight matrix for
subtotal ordering.

Therefore a\textsl{ Mathematica} function invocation of the form\begin{multline*}
\mathtt{TimeConstrained\,}[\mathtt{Timing\,}[\\
\mathtt{GroebnerBasis\,}[\left\{ polynomial_{1},polynomial_{2},\ldots\right\} ,\,\left\{ \mathit{variable}_{1},\ldots,\mathit{variable}_{n}\right\} ,\\
\mathtt{Sort}\rightarrow\mathtt{True},\mathtt{\: MonomialOrder}\rightarrow\mathtt{SubtotalWeightMatrix}[n]];],\\
\mathit{maximumSeconds}]\end{multline*}
computes a subtotal Gröbner basis of the polynomials after heuristically
reordering the variables for speed, then displays only the computing
time or $\mathtt{\$Aborted}$. A similar function invocation using
$\mathtt{DegRevLexWeightMatrix}$ displays the computation time of
a degRevLex basis using a weight matrix.

I didn't want to address the issue of inexact computation at this
time, because it is handled quite differently by different systems.%
\footnote{Many implementations make no special effort for Floats, making the
results disastrously sensitive to differences in floating point arithmetic
and differences in the order of operations.%
} Therefore I avoided test cases that contain Floats, unless they were
obviously representations of rational numbers having small magnitude
denominators, in which case I rationalized those Floats.

I wanted to detect any difference in the relative speeds of term-order
comparisons, and coefficient arithmetic tends to be a larger portion
of the total computing time when Gröbner bases are computed over the
rational numbers rather than over the integers modulo a prime whose
square fits in one computer word. Therefore, I did all problems in
the coefficient domain $\mathbb{Z}_{32003}$, because the prime 32003
maps only about 0.003\% of all non-zero integer coefficients to 0,
but 32003 is small enough so that its square fits within one 32-bit
computer word.

Some of the original test cases and the ones obtained by rationalizing
simple floating-point coefficients contained rational coefficients
that weren't integers. Such polynomials were multiplied by the least
common multiple of their coefficient denominators so that computing
a Gröbner basis in $\mathbb{Z}_{32003}$ was straightforward.

My patience for typing examples, checking for typographical errors,
and waiting for results is limited. Also, I wanted to avoid the extra
space of listing previously unpublished examples in two-dimensional
.pdf format, forcing others to do lengthy error-prone typing to try
all of these examples on some other Gröbner basis implementation.
Therefore I searched the Internet for medium-sized examples -- preferably
available in text that I could copy, paste and quickly edit to replace
semicolons with commas, etc; and I used 120 seconds as the time limit.
Many such examples are available at \cite{Verschelde}. Others are
available at \cite{Posso} and \cite{ginv}. These sites also give
original references for the examples. This article lists input polynomials
for a few additional examples that I think have not previously been
published.

Although the $\mathtt{Sort}\!\rightarrow\mathtt{\! True}$ parameter
should make the results rather insensitive to the order in $\left\{ \mathit{variable}_{1},\ldots,\mathit{variable}_{n}\right\} $,
I also entered the list of variables in the order specified in the
sources or that I could infer, in case it mattered. Some of the original
problems involved eliminating some of the variables or treating some
as parameters in the coefficient domain. However, for uniformity I
simply computed the Gröbner basis with respect to all of the variables.

My objective was to compare the speed of subtotal versus degRevLex
ordering algorithms. To make the comparison fair, I used a weight
matrix for degRevLex too. However, I also used the built-in degRevLex
option to estimate how much improvement to expect if the subtotal
ordering algorithm was built in.

The computer has a 1.60GHz Intel Core 2 Duo U9600 CPU with 3 gigabytes
of RAM.

The Windows Vista operating system appears to have a timer resolution
of only about 0.015 seconds. Therefore if a time was less than 1 second,
then I issued the command\begin{multline*}
\mathtt{Timing\,}[\mathtt{Do}\,[\\
\mathtt{GroebnerBasis\,}[\left\{ \mathit{polynomial}_{1},\mathit{polynomial}_{1},\ldots\right\} ,\left\{ \mathit{variable}_{1},\mathit{variable}_{1},\ldots\right\} ,\\
\mathtt{Sort}\rightarrow\mathtt{True},\:\mathtt{MonomialOrder}\rightarrow\ldots],\{m\}];],\end{multline*}
with the number of repetitions $m$ sufficient to make the time exceed
1 second, then divided by $m$ to obtain the time for one repetition.
Despite this, even with no network connection and no voluntary programs
launched other than\textsl{ Mathematica}, times tend to vary upon
repetition within an interval of about $\pm3\%$ of their mean. \textsl{Mathematica}
uses reference counts rather than garbage collection, so this variation
is probably caused instead by the multiple core architecture and the
irregular competing activity of the many resident programs that lurk
in most Windows installations.

\textsl{Relative} speed is most important for problems that require
extensive time. Therefore of all the problems that I tried, I included
the ones that took the most time for built-in degRevLex order without
exceeding the time limit for any ordering algorithm -- as many examples
as fit in a one-page table. The examples that I thus excluded didn't
have noticeably different overall behavior regarding the relative
speed of subtotal \emph{versus} degRevLex ordering algorithms.

\subsection{The test examples}

Here are the included examples that I believe aren't already publicly
published in some form:
\begin{enumerate}
\item Lichtblau 1:%
\footnote{The problems provided by Daniel Lichtblau are from a mix of literature,
user questions and bug reports, with unrecorded individual provenance.%
}\[
\begin{array}{c}
t^{4}zb+x^{3}ya,\\
tx^{8}yz-ab^{4}cde,\\
xy^{2}z^{2}d+zc^{2}e^{2},\\
tx^{2}y^{3}z^{4}+ab^{2}c^{3}e^{2},\end{array}\]
with variable order $\left\{ t,x,y,z,a,b,c,d,e\right\} $.
\item Lichtblau 2:\[
\begin{array}{c}
a^{2}+b^{2}+2c^{2}+2d^{2}+3f^{2}+3g^{2}-h,\\
70\, ab+140\, ac+140\, bd+28\, cd+252\, cf+252\, dg+18fg-105\, u,\\
28\, bc+28\, c^{2}+28\, ad+28\, d^{2}+42\, af+12\, df+24f^{2}+42\, bg+12\, cg+24\, g^{2}-35\, v,\\
36\, cd+30\, bf+24\, cf+30\, ag+24\, dg+16fg-35\, x,\\
8\, df+2f^{2}+8cg+2g^{2}-7y,\\
100fg-77\, z,\end{array}\]
with variable order $\left\{ a,b,c,d,f,g,h,u,v,x,y,z\right\} $.
\item Lichtblau 3:\[
\begin{array}{c}
-3375\, uv+3291\, u^{2}v+750\, uv^{2}-732\, u^{2}v^{2}+2225\, uv^{3}-2209\, u^{2}v^{3}+Qx,\\
P+Qy,\\
P+Qz,\\
r(1+24\, uv+24\, u^{2}v-24\, uv^{2}),\end{array}\]
where\begin{eqnarray*}
P\!\! & \!=\! & \!\!-400\, u\!+\!350\, u^{2}\!-\!4800\, uv\!+\!4800\, u^{2}v\!+\!7425\, uv^{2}\!-\!7359\, u^{2}v^{2}\!-\!2225\, uv^{3}\!+\!2209\, u^{2}v^{3},\\
Q\!\! & \!=\! & \!\!(1+24\, uv+24\, u^{2}v-24\, uv^{2}),\end{eqnarray*}
with variable order $\left\{ x,y,z,u,v,r\right\} $.
\item A geometry problem of Michael Trott:\[
\begin{array}{c}
-x_{1}+x_{2}+y_{1}-3x{}_{1}^{2}y_{1}+2x{}_{1}^{3}y_{1}-2x_{1}y{}_{1}^{3}+2x_{2}y{}_{1}^{3}-y_{2}+3x{}_{1}^{2}y_{2}-2x{}_{1}^{3}y_{2},\\
x_{1}-x_{2}-y_{1}+3x{}_{2}^{2}y_{1}-2x{}_{2}^{3}y_{1}+y_{2}-3x{}_{2}^{2}y_{2}+2x{}_{2}^{3}y_{2}+2x_{1}y{}_{2}^{3}-2x_{2}y{}_{2}^{3},\\
-1+2x_{1}-2x{}_{1}^{3}+x{}_{1}^{4}+2y_{1}+y{}_{1}^{4},\\
-1+2x_{2}-2x{}_{2}^{3}+x{}_{2}^{4}+2y_{2}+y{}_{2}^{4},\\
1+(x_{1}-x_{2})z+(y_{1}-y_{2})z^{2},\end{array}\]
with variable order $\left\{ x_{1},x_{2},y_{1},y_{2},z\right\} $.
\item \textsl{Mathematica} \textsf{Help\,-\,GröbnerBasis\,-\,Options\,-\,Sort}:\[
\begin{array}{c}
3x^{7}+5xyz^{2}-10\, y^{2}z-6xz+y^{3}+w,\\
-2x^{2}z+3x^{3}y^{2}+y^{4}-12\, xz-8xz^{2}+3y^{2}z-11\, wxy^{2},\\
10\, x^{2}w-7yzw^{2}-2xz^{4}w+4x^{2}y+3xy^{2}-6yz^{3}-w+2,\\
w^{3}-wx^{2}y+xyz^{2}-2wxz^{2}-3w-2xy^{2}-3,\end{array}\]
with variable order $\left\{ w,x,y,z\right\} $.
\item Variation on a theme of Giovini \emph{et}. \emph{al}: One of my test
files is so similar to Giovini 3.7 \cite{GioviniEtAl} that I must
have entered it twice, but once with typographical errors. These seemingly
minor changes make the example that is relatively nearly fastest for
subtotal ordering algorithm become an example that is relatively slowest
for that algorithm. The variation is\[
\begin{array}{c}
x^{33}z^{23}-y^{82}a,\\
x^{45}-y^{13}z^{21}b,\\
x^{41}c-y^{33}z^{12},\\
x^{22}-y^{33}z^{12}d,\\
x^{5}y^{17}z^{22}e-1,\\
xyzt-1,\end{array}\]
with variable order $\left\{ t,b,c,e,d,a,z,x,y\right\} $.
\end{enumerate}
All of the other examples are already published elsewhere, as cited
in Table 1.

\subsection{Test results}

The last column of Table 1 displays the most important results --
the ratios of the computing time for subtotal ordering versus degRevLex
ordering, both using a weight matrix.

Notice that the rows are ordered by non-decreasing values in this
last column in an attempt to discern correlations with the number
of variables and/or total degrees of the input polynomials. These
total degrees are also listed in the Table with, for example, $6^{3}\!\cdot\!5$
meaning 3 polynomials each having total degree 6 and 1 polynomial
having total degree 5. However, the number of examples and/or their
variability doesn't appear to be large enough to reveal obvious correlations.
This might be partly because the standard deviation of the number
of variables is only 2.8 for a median and mean of about 8.

Most of the speed ratios are rather close to 1.0, as expected. However,
there are some notable outliers at at the top and bottom of the table.
A probable explanation for these outliers is that in\textsl{ Mathematica}
the weight vectors are also used to select the next S-polynomial or
reducer polynomial. If so, then it is an indication that good selection
of a next S-polynomial and a next reducer polynomial should instead
if possible be based on the \textsl{ordering} that the weight vector
induces, so as to be invariant to a property that doesn't correlate
completely with ordering.

Although re-executing an example caused time variations for individual
examples within an interval of about $\pm3\%$, there are enough examples
so that the summary statistics are more tightly repeatable.

The median speed ratio of 0.98 and mean of 0.92 suggest that the subtotal
algorithm would be about this much faster than the degRevLex algorithm
if both were built into\textsl{ Mathematica}. The standard deviation
of 0.24 weakens that conclusion, but most of the standard deviation
comes from the outliers at the top of the table where the subtotal
algorithm was significantly faster.

The speeds of the two methods are close enough so that quite possibly
the degRevLex algorithm could be faster than the subtotal algorithm
if both were built into another Gröbner basis implementation that
used weight vectors in a different way to select S polynomials and
reducers. However, it seems highly likely that the speeds would be
quite close if weight vectors weren't used for selection strategy.

The penultimate column of Table 1 is next most interesting:
\begin{itemize}
\item The summary statistics for that column weakly support the conclusion
that built-in degRevLex is faster than the degRevLex weight matrix
by only 2\% for the median or 8\% for the mean. This can only be attributable
to computing the initial weight vectors from the weight matrices,
and these small percentages indicate that the initialization of weight
vectors is usually only a small portion of the total time.
\item However, there are a surprising number of instances where the weight
matrix is slightly \textsl{faster} -- up to 7\%. I can think of no
implementation reason why this should be. If there is no such reason,
then it can be taken as an indication of the repeatability deviations
in \textsl{individual} time ratios -- up to 7\% rather than the 3\%
that I estimated from informal experiments.
\end{itemize}
\begin{table}[h]
\centering{}\caption{Coefficient in $\mathbb{Z}_{32003}$. Seconds \& time ratios for subtotal
\emph{vs} degRevLex order}
\begin{tabular}{|c|c|c|c|c|c|}
\hline 
search term \& citations & ${\#\atop \mathrm{vars}}$ & input tot. degrees & $\underset{\mathrm{sec.}}{{\mathrm{grevlex}\atop \mathrm{builtin}}}$ & $\dfrac{\mathrm{grevlex}}{\left({\mathrm{grevlex}\atop \mathrm{matrix}}\right)}$ & $\dfrac{\mathrm{subtotal}}{\left({\mathrm{grevlex}\atop \mathrm{matrix}}\right)}$\tabularnewline
\hline
\hline 
Cohn3 \cite{ginv} & 4 & $6^{3}\!\cdot\!5$ & 7.24 & 0.08 & \textbf{0.08}\tabularnewline
\hline 
filter design \cite{Posso} & 9 & $5\!\cdot\!,4^{2}\!\cdot\!3\!\cdot\!,2^{4}$ & 10.1 & 0.98 & \textbf{0.22}\tabularnewline
\hline 
benchmark\_i1 \cite{ginv} & 10 & $3^{10}$ & 2.78 & 0.30 & \textbf{0.31}\tabularnewline
\hline 
Assur44 \cite{ginv} & 8 & $3^{3}\!\cdot\!2^{5}$ & 6.30 & 0.40 & \textbf{0.38}\tabularnewline
\hline 
Giovini 3.7 \cite{GioviniEtAl}  & 9 & $83\!\cdot\!45^{3}\!\cdot\!44\!\cdot\!4$ & 36.7 & 0.95 & \textbf{0.65}\tabularnewline
\hline 
benchmark\_D1 \cite{ginv} & 12 & $3^{2}\!\cdot\!2^{9}\!\cdot\!1$ & 0.71 & 0.59 & \textbf{0.73}\tabularnewline
\hline 
des22\_24 \cite{ginv,Verschelde} & 10 & $2^{8}\!\cdot\!1^{2}$ & 0.75 & 0.70 & \textbf{0.73}\tabularnewline
\hline 
Lichtblau 2 & 9 & $11\!\cdot\!10\!\cdot\!6^{2}$ & 0.44 & 0.72 & \textbf{0.90}\tabularnewline
\hline 
Gonnet \emph{et}. \emph{al}. \cite{GonnetEtAl} & 17 & $2^{19}$ & 6.01 & 0.74 & \textbf{0.95}\tabularnewline
\hline 
cdpm5 \cite{ginv} & 5 & $3^{5}$ & 4.60 & 0.95 & \textbf{0.95}\tabularnewline
\hline 
reimer5 \cite{Verschelde} & 5 & $6\!\cdot\!5\!\cdot\!4\!\cdot\!3\!\cdot\!2$ & 1.59 & 0.99 & \textbf{0.96}\tabularnewline
\hline 
Kotsireas4body \cite{Posso} & 6 & $5^{3}\!\cdot\!2^{3}$ & 2.92 & 1.05 & \textbf{0.97}\tabularnewline
\hline 
Lichtblau 3 & 12 & $2^{6}$ & 0.50 & 0.88 & \textbf{0.97}\tabularnewline
\hline 
cyclic6 \cite{ginv,Verschelde} & 6 & $6\!\cdot\!5\!\cdot\!4\!\cdot\!3\!\cdot\!2,\!\cdot\!1$ & 0.62 & 1.01 & \textbf{0.97}\tabularnewline
\hline 
Giovini 3.1 \cite{GioviniEtAl} & 7 & $4^{2}\!\cdot\!3^{10}\!\cdot\!2$ & 0.57 & 1.17 & \textbf{0.97}\tabularnewline
\hline 
eco8 \cite{ginv} & 8 & $3^{3}\!\cdot\!2\!\cdot\!1$ & 1.18 & 0.95 & \textbf{0.98}\tabularnewline
\hline 
redeco7 \cite{Verschelde} & 8 & $2^{6}\!\cdot\!1^{2}$ & 0.67 & 0.96 & \textbf{0.98}\tabularnewline
\hline 
extcyc5 \cite{ginv} & 6 & $5^{2}\!\cdot\!4\!\cdot\!3\!\cdot\!2\!\cdot\!1$ & 1.15 & 1.00 & \textbf{0.98}\tabularnewline
\hline 
f744 \cite{ginv} & 12 & $3^{2}\!\cdot\!2^{2}\!\cdot\!1^{2}$ & 5.13 & 0.93 & \textbf{0.98}\tabularnewline
\hline 
Lichtblau 1 & 6 & $5^{3}$ & 5.91 & 1.02 & \textbf{0.99}\tabularnewline
\hline 
virasoro \cite{Verschelde} & 8 & $2^{8}$ & 23.8 & 0.99 & \textbf{1.00}\tabularnewline
\hline 
Trott geometry & 5 & $4^{3}\!\cdot\!3$ & 12.4 & 0.99 & \textbf{1.00}\tabularnewline
\hline 
Kotsireaus4bodySymmetric \cite{Posso} & 7 & $5^{2}\!\cdot\!3^{2}\!\cdot\!2^{2}$ & 36.2 & 1.00 & \textbf{1.01}\tabularnewline
\hline 
Katsura7 \cite{Verschelde} & 7 & $2^{6}\!\cdot\!1$ & 0.78 & 0.99 & \textbf{1.01}\tabularnewline
\hline 
redcyc6 \cite{Verschelde} & 6 & $11\!\cdot\!5\!\cdot\!4\!\cdot\!3\!\cdot\!2\!\cdot\!1$ & 0.43 & 1.00 & \textbf{1.01}\tabularnewline
\hline 
Harrier RK2 \cite{BoegeGebauerAndKredel} & 13 & $4^{4}\!\cdot\!3^{2}\!\cdot\!2\!\cdot\!1^{4}$ & 33.5 & 0.95 & \textbf{1.03}\tabularnewline
\hline 
\textsl{Mathematica} help & 4 & $7\!\cdot\!6\!\cdot\!5\!\cdot\!4$ & 5.98 & 1.01 & \textbf{1.04}\tabularnewline
\hline 
rpb124 \cite{Verschelde} & 9 & $3^{2}\!\cdot\!2^{6}\!\cdot\!1$ & 1.87 & 1.07 & \textbf{1.04}\tabularnewline
\hline 
Tran, rational implicitization \cite{Tran} & 5 & $6^{2}\!\cdot\!5$ & 5.13 & 0.96 & \textbf{1.05}\tabularnewline
\hline 
kinema \cite{ginv} & 9 & $2^{6}\!\cdot\!1^{3}$ & 2.07 & 1.03 & \textbf{1.07}\tabularnewline
\hline 
Kotsireaus5body \cite{Posso,Verschelde} & 6 & $5^{3}\!\cdot\!2^{3}$ & 3.18 & 1.03 & \textbf{1.10}\tabularnewline
\hline 
rpbl \cite{Verschelde}  & 6 & $3^{5}\!\cdot\!2$ & 0.94 & 1.07 & \textbf{1.13}\tabularnewline
\hline 
variation on Giovini 3.7 & 9 & $83\!\cdot\!46\!\cdot\!45^{3}\!\cdot\!4$ & 19.8 & 0.94 & \textbf{1.39}\tabularnewline
\hline 
\textbf{Statistics$\downarrow$} &  &  &  &  & \tabularnewline
\hline 
\textbf{Median} & 8.0 &  & 2.85 & 0.98 & \textbf{0.98}\tabularnewline
\hline 
\textbf{Mean} & 8.1 &  & 7.34 & 0.92 & \textbf{0.92}\tabularnewline
\hline 
\textbf{Standard Deviation} & 2.8 &  & 10.7 & 0.19 & \textbf{0.24}\tabularnewline
\hline 
\textbf{Count < 1.0000} &  &  &  & 22 & \textbf{20}\tabularnewline
\hline 
\textbf{Count > 1.0000} &  &  &  & 10 & \textbf{11}\tabularnewline
\hline
\end{tabular}
\end{table}

\section{\label{sec:Conclusions}Conclusions}

The main result is that subtotal ordering is an alternate way to view
degRevLex ordering that more clearly explains its good behavior.

A secondary conclusion is that the the algorithm for computing subtotal
order is very nearly the same speed as that for computing degRevLex
order -- at least if the induced ordering rather than the weight vectors
is used to select the next S-polynomial and next reducer.

The relative speeds of the subtotal and degRevLex ordering algorithms
are close enough so that it probably isn't worth replacing the degRevLex
algorithm with the subtotal algorithm in existing implementations.
However, it is well worth considering use of the subtotal algorithm
instead of the degRevLex algorithm in new implementations.

\section*{Acknowledgments}

Thank you Daniel Lichtblau for your extensive patient help and encouragement.


\begin{thebibliography}{10}
\bibitem{ginv}Blinkov, Yu.A. and Gerdt, V.P., 2006. GINV polynomial
test suite,\\
 \url{http://invo.jinr.ru/ginv/index.html}

\bibitem{BuchbergerThesisGerman}Buchberger, B., 1965. Ein Algorithmus
zum Auffinden der Basiselemente des Restklassenringes nach einem nulldimensionalen
Polynomideal, Mathematical Institute, University of Innsbruck, Austria.
PhD Thesis.

\bibitem{BuchbergerThesisEnglish}Buchberger, B., 2006. An Algorithm
for Finding the Basis Elements in the Residue Class Ring Modulo a
Zero Dimensional Polynomial Ideal, an English translation of \cite{BuchbergerThesisGerman},
\textsl{Journal of Symbolic Computation} 41, (3-4), pp. 475-511. 

\bibitem{GioviniEtAl}Giovini, A., Mora, T., Niesi, G., Robbiano,
L., and Traverso, C., 1991: {}``One sugar cube please'', or Selection
strategies in the Buchberger algorithm, \textsl{Proceedings of ISSAC
1991}, pp. 49-54.

\bibitem{GonnetEtAl}Gonnet, G.H., Char, B.W., Geddes, K.O., 1983.
Solution of a general system of equations, \textsl{ACM SIGSAM Bulletin}
17 (3 \& 4), pp. 48-49.

\bibitem{Posso}D. Bini \& B. Mourrain, 2011. Polynomial test suite,\\
 \url{http://www-sop.inria.fr/saga/POL/} 

\bibitem{Tran}Tran, Q.N., 2004. Efficient Gröbner walk conversion
for implicitization of geometric objects, \textsl{Computer Aided Geometric
Design} 21 (9), pp. 837-857.

\bibitem{Trinks}Trinks, W., 1978. Über B. Buchbergers verfahren,
systeme algebraischer gleichungen zu lösen, \textsl{Journal of Number
Theory} 10, pp. 475-488.

\bibitem{Verschelde}Verschelde, J., 2011. A collection of examples
for polynomial systems,\\
 \url{http://www.math.uic.edu/~jan/Demo/CHARAC.html}

\bibitem{Yap}Yap, Chee-Keng, 1999. \textsl{Fundamental Problems in
Algorithmic Algebra}, Oxford University Press, ISBN 0-19-512516-9,
Lecture XII.
\end{thebibliography}
\end{document}